\documentclass[reqno,12pt]{amsart}
\usepackage{amsthm,setspace}
\usepackage[usenames,dvipsnames]{xcolor}
\usepackage{amsmath}
\usepackage{amssymb,graphicx,bm,subfigure}
\usepackage{hyperref,cancel}
\usepackage{tikz}
\usepgflibrary{fpu}
\usepackage{pgf}
\usepackage{pgffor}

\numberwithin{equation}{section}
\newcommand{\E}{\mathcal{E}}
\newcommand{\U}{\mathcal{U}}
\newcommand{\W}{W}
\newcommand{\x}{\mathbf{x}}
\newcommand{\f}{\mathbf{f}}
\renewcommand{\u}{\mathbf{u}}
\newcommand{\y}{\mathbf{y}}
\renewcommand{\v}{\mathbf{v}}
\newcommand{\n}{\mathbf{n}}
\newcommand{\R}{\mathbb{R}}

\newcommand{\abs}[1]{\lvert #1 \rvert}

\renewcommand{\b}{{\mathsf b}}

\newtheorem{theorem}{\bf Theorem}[section]

\newtheorem{corollary}{\bf Corollary}[section]
\newtheorem{definition}{\bf Definition}[section]
\newtheorem{lemma}{\bf Lemma}[section]
\newtheorem{remark}{\bf Remark}[section]

\begin{document}

\title{On the Equivalence of Local and Global Area-Constraint
  Formulations for Lipid Bilayer Vesicles}
\author{Sanjay Dharmavaram$^{1}$ and Timothy J. Healey$^{2}$} 
\address{$^{1}$Dept. of Mechanical and Aerospace Engineering\\
Cornell University, Ithaca, NY.\\
$^{2}$Dept. of Mathematics\\
Cornell University, Ithaca, NY.\\
}



\begin{abstract}
Lipid bilayer membranes are commonly modeled as area-preserving fluid
surfaces that resist bending. There appear to be two schools of
thought in the literature concerning the actual area constraint. In
some works the total or global area ($GA$) of the vesicle is a
prescribed constant, while in others the local area ratio is assigned
to unity. In this work we demonstrate the equivalence of these
ostensibly distinct approaches in the specific case when the
equilibrium configuration is a smooth, closed surface of genus
zero. We accomplish this in the context of the Euler-Lagrange
equilibrium equations, constraint equations and the second-variation
with admissibility conditions, for a broad class of models -- including
the phase-field type.
\end{abstract}

\maketitle
\section{Introduction}
Lipid membranes are commonly modeled as area-preserving fluid surfaces
that resist bending. There appear to be two schools of thought in the
literature concerning the actual area constraint. As in the pioneering
work of Helfrich \cite{helfrich}, the total or global area ($GA$) of
the vesicle is often a prescribed constant, e.g.,
cf. \cite{bonito2010parametric}, \cite{elliott2010modeling},
\cite{taniguchi}. On the other hand, in analogy with the assumption of
incompressibility in 3D continuum mechanics, it is natural to consider
2D-surface models characterized by local area ($LA$) preservation,
i.e., the local area ratio is assigned to unity. In this work we
demonstrate the equivalence of these ostensibly distinct
approaches in the specific case when the equilibrium configuration is
a smooth, closed surface of genus zero. We accomplish this in the
context of the Euler- Lagrange equilibrium equations, constraint
equations and the second-variation with admissibility conditions, for a
broad class of models -- including the phase-field type.

The outline of this work is as follows. In Section 2 we summarize the
two formulations -- $LA$ and $GA$ -- identifying their essential
differences, due to distinct constraints. We note the
reparametrization symmetry inherent in the $GA$ formulation in
Section 3. This leads naturally to the definition of an equivalence
class of solutions, each member of which has the same total surface
area while satisfying the $GA$ field equations. In Section 4 we show
there exists an element of the equivalence class that preserves the
local area ratio, as a deformation of the unit sphere $S^2$. Our result
here depends crucially upon the Riemann-Roch theorem, which insures
that a closed, genus-zero surface in $\R^3$ is conformally equivalent
to $S^2$, cf. \cite{Jost_CRS}. With this in hand we construct an
explicit change of coordinates ultimately yielding a locally
area-preserving solution.  

In Section 5 we consider the second-variation conditions, normally
associated with the determination of local energy minima. While the
expressions for the second variation agree, the area constraint
equations for the two formulations again yield apparently
different admissibility conditions. In particular, the
pointwise condition for the $LA$ formulation involves the tangential
variation, while the integral condition for the $GA$ formulation is
independent of it. Employing the Hodge decomposition
theorem -- in this case on a genus-zero, closed surface in $\R^3$ -- we
are able to show that the pointwise condition for the former reduces
to the integral condition of the latter.

\section{Formulations}
\label{sec:formulation}
Without loss of generality, we take the unit sphere $S^2$ as the
reference surface of the lipid vesicle. The deformed surface of the
vesicle is denoted by $\Sigma$. The \emph{deformation},
$\f:S^2\rightarrow \Sigma \subset \R^3$, relates the reference and the
current configurations of the membrane, \emph{i.e.,}
$\Sigma=\f(S^2)$. In the following we employ the Einstein convention
for tensors. Greek letters will be used to represent tensor indices
which are assumed to lie in the set $\{1,2\}$ corresponding to surface
coordinates, say, $\x:=(x^1,x^2)$ on $S^2$.

We consider a general class of phase-field  models for multi-phase
lipid membrane vesicles, with total internal potential energy given by
\begin{equation}
\U(\f,\phi)=\int_\Sigma \W(H,K,\phi,\abs{\nabla\phi})\;da,
\end{equation}
where $\W$ represents the energy density as a function of the mean and
Gaussian curvatures fields, $H$ and $K$, respectively, on $\Sigma$,
and $da$ represents the area measure on $\Sigma$. The phase field
variable $\phi$ represents the normalized difference in concentration
of the lipid components. Here $\nabla\phi$ refers to the (spatial)
gradient of $\phi$ on $\Sigma$, cf. \cite{taniguchi}. When $\phi\equiv
0$, we recover a generalized Helfrich model \cite{helfrich}.

Since the two components do not react with each other, their total
concentration is fixed on $\Sigma$. Accordingly, we impose the
constraint,
\begin{equation}
\int_\Sigma (\phi-\mu)\;da = 0,
\label{eq:conc_constr}
\end{equation}
where $\mu$ represents a fixed average concentration.

There are two common approaches to model area preservation in lipid
membranes -- local and global. In the former, every infinitesimal patch
on the surface of the membrane preserves area under deformation
\cite{jenkinsRBC,jenkins1977,steigman}. In the other approach, the
total area of the membrane is presumed fixed under deformation
\cite{Elliot,bonito2010parametric,fengKlug,maKlug}.
\subsection{Local Area (LA) Constrained Formulation}
\label{sc:LA_formulation}

In this formulation, area preservation is modeled locally via
\begin{equation}
J\equiv 1,
\label{eq:LA_constr}
\end{equation}
where the local area ratio $J$ is defined by
$J^2:=\det(D\f^TD\f)$, with $D\f$ denoting the total surface derivative of $\f$.

The total potential energy can be expressed as
\begin{equation}
\E_{LA}(\u) = \int_\Sigma \W(H,K,\phi,\abs{\nabla\phi})\;da
+\int_{S^2}\gamma_L(\x)(J-1)\;dA +\lambda\int_\Sigma (\phi-\mu)\;da-pV,
\label{eq:energy_LA}
\end{equation}
where $\gamma_L(\x)$ is the Lagrange multiplier field enforcing the
area constraint (\ref{eq:LA_constr}), $\lambda$ is the Lagrange
multiplier associated with the constraint on the phase field
(\ref{eq:conc_constr}), $p$ is the internal excess pressure and $V$ is
the total volume enclosed by $\Sigma$. For notational convenience we
write $\u:=(\f,\phi)$. 

The Euler-Lagrange equations are obtained by taking variations of the
energy (\ref{eq:energy_LA}) with respect to the fields $\f$ and
$\phi$. The variation of the deformation $\f$ is
defined by
\begin{equation}
\f(\x) \mapsto \f(\x) +\alpha\bm{\eta}(\x).
\label{eq:f->f+eta}
\end{equation}
for sufficiently small $\alpha$.

It is standard in this field is to express the variation ${\bm
  \eta}$ in terms of its components normal and tangential to the
surface $\Sigma$. The variation $\bm{\eta}$ when pushed forward using
$\y=\f(\x)$ takes the form
\begin{subequations}
\begin{equation}
 \y \mapsto \y + \alpha[\v(\y)+w(\y)\n(\y)],
\label{eq:def_var_cur}
\end{equation}
where $\v$ and $w$ are the tangential and normal variations expressed
in the current configuration, and $\n$ is the unit normal field on
$\Sigma$. The total variation of the phase field is then of the form
\begin{equation}
\phi(\y)\mapsto \phi(\y) + \alpha[\nabla\phi(\y)\cdot\v(\y) + \psi(\y)],
\label{eq:conc_var_cur}
\end{equation}
\end{subequations}
where $\psi$ is the spatial variation. In this way, the first
variation condition can be expressed as
\begin{multline}
\delta\E_{LA} = \int_\Sigma \Big\{\frac{1}{2}\Delta \W_H
+\tilde{\Delta}\W_K +
\frac{\W_{\Phi}}{\abs{\nabla\phi}}\b[\nabla\phi,\nabla\phi]+\\(2H^2-K)\W_H+2KH\W_K
-2H\Big[\W+\gamma_L+\lambda(\phi-\mu)\Big]-p\Big\}w+\\ \nabla\gamma_L\cdot\v+\Big\{
-\nabla\cdot\Big(\frac{\W_{\Phi}}{\abs{\nabla\phi}}\nabla\phi\Big)
+ \W_\phi + \lambda\Big\} \psi\;da=0,
\label{eq:LA_variation}
\end{multline}
for all smooth variations $\v$, $w$ and $\psi$, where
$\nabla\cdot(\cdot)$ is the surface divergence and $\Delta(\cdot)$ is
the Laplace-Beltrami on $\Sigma$, the various partial derivatives of
$\W(H,K,\phi,\Phi)$ are denoted $\W_H=\frac{\partial\W}{\partial H}$,
etc., and $\tilde\Delta(\cdot) :=
\tilde{b}^{\alpha\beta}\nabla_\alpha\nabla_\beta(\cdot)$ (using the
notation of \cite{tu2004geometric}), $\tilde{b}^{\alpha\beta}$ being
the cofactor matrix of the second fundamental form
$\b[d\x,d\x]:=b_{\alpha\beta}dx^\alpha dx^\beta$ of
$\Sigma$, \emph{i.e.,} $\b=-\nabla\n$.

Setting the variations $w$, $\v$ and $\psi$ pair-wise to be
identically equal to zero, we obtain the following Euler-Lagrange
equations,
\begin{subequations}
\label{eq:EL_L}
\begin{multline}
\frac{1}{2}\Delta \W_H +\tilde{\Delta}\W_K + \frac{\W_{\Phi}}{\abs{\nabla\phi}}\b[\nabla\phi,\nabla\phi]+(2H^2-K)\W_H+2KH\W_K
\\-2H\Big[\W+\gamma_L+\lambda(\phi-\mu)\Big]-p=0,
\label{eq:EL_L_normal}
\end{multline}
\begin{equation}
\nabla \gamma_L = {\mathbf 0},
\label{eq:EL_L_tangential}
\end{equation}
\begin{equation}
-\nabla\cdot\Big(\frac{\W_{\Phi}}{\abs{\nabla\phi}}\nabla\phi\Big) + \W_\phi + \lambda = 0.
\label{eq:EL_L_phi}
\end{equation}
We recover the  constraint equations by taking variations of the
energy with respect to the Lagrange multipliers:
\begin{equation}
J\equiv 1,
\label{eq:EL_L_constraint}
\end{equation}
\begin{equation}
\int_\Sigma (\phi-\mu)\;da =0.
\end{equation}
\end{subequations}
Of course (\ref{eq:EL_L_tangential}) implies that $\gamma_L$ is a
constant over $\Sigma$.
\subsection{Global Area (GA) Constrained Formulation}

In this formulation, the total surface area of
the membrane is assumed to be constant, \emph{viz.},
\begin{equation}
\int_\Sigma da = 4\pi.
\label{eq:GA_constr}
\end{equation}
The total energy here is similar to that for the $LA$ formulation
(discussed above), except for the term involving the area constraint,
now  associated with a scalar Lagrange multiplier $\gamma_G$:
\begin{equation}
\E_{GA} = \int_\Sigma \W(H,K,\phi,\abs{\nabla\phi})\;da +\gamma_G \int_\Sigma\;da +\lambda\int_\Sigma (\phi-\mu)\;da-pV.
\label{eq:energy_global}
\end{equation}
The first variation condition of $\E_{GA}$ is identical to
(\ref{eq:LA_variation}), except the term involving the tangential
variation vanishes identically. As in the $LA$ case, we then take the
normal and phase-field variations and obtain the following
Euler-Lagrange equations, respectively:
\begin{subequations}
\label{eq:EL_G}
\begin{multline}
\frac{1}{2}\Delta \W_H +\tilde{\Delta}\W_K + \frac{\W_{\Phi}}{\abs{\nabla\phi}}\b[\nabla\phi,\nabla\phi]+(2H^2-K)\W_H+2KH\W_K
\\-2H\Big[\W+\gamma_G+\lambda(\phi-\mu)\Big]-p=0,
\label{eq:EL_G_normal}
\end{multline}
\begin{equation}
-\nabla\cdot\Big(\frac{\W_{\Phi}}{\abs{\nabla\phi}}\nabla\phi\Big) + \W_\phi + \lambda = 0.
\label{eq:EL_G_phi}
\end{equation}
The associated constraints are
\begin{equation}
\int_\Sigma da = 4\pi,
\label{eq:EL_G_constraint}
\end{equation}
\begin{equation}
\int_\Sigma (\phi-\mu)\;da =0.
\end{equation}
\end{subequations}

Note that unlike the $LA$ formulation, the tangential equation
vanishes identically in this case. We further observe that since
$\gamma_{L}$ is constant on $\Sigma$ (cf. Section
2\ref{sc:LA_formulation}) the system of partial differential equations
-- (\ref{eq:EL_L}a,c) and (\ref{eq:EL_G}a,b) -- for the two
formulations are identical if we set $\gamma_L=\gamma_G=\gamma$.  In
the absence of the phase field $\phi$, this observation has been noted
in other works \cite{jenkinsRBC,steigmanFluidFilms}. Although it may
be tempting to deduce from this that the two formulations are
equivalent, this argument alone is insufficient, due to the fact that
the area constraints, (\ref{eq:EL_L_constraint}) and
(\ref{eq:EL_G_constraint}), are different. Of course, any solution of
the $LA$ formulation is also a solution of the $GA$ formulation.

\section{Reparametrization Symmetry}
\label{sec:repara_sym}
The vanishing of the tangential equation noted above can be attributed
to the reparametrization symmetry of the $GA$ formulation. This
infinite-dimensional symmetry group is a manifestation of the fluidity
of lipid membranes. In this section we see that as a consequence of
the symmetry, solutions of the $GA$ formulation (when they exist)
belong to an infinite-dimensional equivalence class. To proceed, we need
to be more precise: 

By a \emph{solution} of the $LA$ formulation ($GA$ formulation), we
mean there are smooth mappings $\f:S^2\to \Sigma\subset \R^3$ and
$\phi:\Sigma\to\R$, and a smooth parametrization of $S^2$, ${\bf
X}:\Omega\to S^2$, where $\Omega\subset \R^2$ is a coordinate chart,
such that
\begin{subequations}
\begin{equation}
{\bf Y}(X^1,X^2):=\f({\bf X}(X^1,X^2))
\label{eq:solution_Y}
\end{equation}
and
\begin{equation}
\Psi(X^1,X^2):=\phi({\bf Y}(X^1,X^2))
\end{equation}
\label{eq:solution}%
\end{subequations}
satisfy the system (\ref{eq:EL_L}) ((\ref{eq:EL_G})) identically. In particular,
(\ref{eq:EL_L_constraint}) reads
\begin{subequations}
\begin{equation}
J=\sqrt{\frac{a}{A}}\equiv 1,
\label{eq:J_equals_1}
\end{equation}
where
\begin{equation}
a:=\det[a_{\alpha\beta}],
\label{eq:a}
\end{equation}
\begin{equation}
A:=\det[A_{\alpha\beta}],
\end{equation}
with
\begin{equation}
a_{\alpha\beta}:={\bf Y}_{,\alpha}\cdot {\bf Y}_{,\beta},
\end{equation}
\begin{equation}
A_{\alpha\beta}:={\bf X}_{,\alpha}\cdot {\bf X}_{,\beta}.
\end{equation}
\label{eq:a_and_A}%
\end{subequations}
being the components of the first fundamental forms for $\Sigma$ and
$S^2$, respectively, where ${\bf Y}_{,\alpha}:=\frac{\partial{\bf
    Y}}{\partial X^\alpha}$, \emph{etc}. In what follows we consistently employ
the above convention, \emph{viz.}, the determinant of the matrix of
components of the first fundamental form is denoted by the same letter
employed (with indices) for the components.

Clearly, the $GA$ formulation (\ref{eq:energy_global}) is independent
of the coordinate parametrization of $\Sigma$, reflecting the in-plane
fluidity of the model. Therefore, for any $\chi\in\text{Diff}(S^2)$,
the diffeomorphism group of $S^2$ into itself, we have
\begin{equation}
\E_{GA}(\u(\x))=\E_{GA}(\u\circ\chi(\x)).
\label{eq:repara}
\end{equation}
It then follows from (\ref{eq:repara}) and the first variation
condition that if $\u$ is a solution to the system of equations
(\ref{eq:EL_G}) then so is $\u\circ\chi$. This observation motivates
the following definition.

\begin{definition}[Equivalence class of solutions]
Two solutions $\u$ and $\u^*$ of $GA$ formulation are said to be
equivalent, if there exists a $\chi\in\text{Diff}(S^2)$ such that
\begin{equation}
\u^*=\u\circ\chi = (\f\circ\chi,\phi(\f\circ\chi)).
\label{eq:equiv_class}
\end{equation}
We denote this equivalence class by $[\u]$. Clearly, this set
contains infinite elements, presuming the existence of a solution.
\end{definition}
\noindent
We remark that $\text{Diff}(S^2)$ is not a symmetry group of the $LA$
formulation, since the local area constraint
(\ref{eq:EL_L_constraint}) under reparametrizations $\chi$ transforms
as
\begin{equation}
J(\chi(\x))\det(D\chi(\x));
\label{eq:Jdet}
\end{equation}
(\ref{eq:EL_L_constraint}) is clearly not
invariant unless $\det(D\chi(\x))\equiv 1$.

\section{Equivalence: Equilibria}
\label{sec:equiv_equil}
Let us first recall that any solution of the $LA$ formulation is
automatically a solution of the $GA$ formulation and therefore, the
solution set of the former is a subset of the latter. In this section,
we establish a converse by showing that there is a representative in
the equivalence class of the genus-zero solutions of the $GA$
formulation that satisfies the local area constraint. By explicit
construction, we show that any solution $\u(\x)$ of the $GA$
formulation may be mapped to a solution of the $LA$ formulation via a
diffeomorphism. Since the Euler-Lagrange equations and the constraint
on phase field are identical for the two formulations, it is
sufficient to consider the area constraints only.

We first consider the usual parametrization of $S^2$ via spherical coordinates,
$(\theta^1,\theta^2)\in \Omega:=(0,\pi)\times (0,2\pi)$, \emph{viz.},
\begin{equation}
{\bf R}(\theta^1,\theta^2) = \sin\theta^1\Big(\cos\theta^2 {\bf e}_1+\sin\theta^2 {\bf e}_2\Big)+\cos\theta^1 {\bf e}_3,
\label{eq:R}
\end{equation}
where $\{{\bf e}_1,{\bf e}_2, {\bf e}_3\}$ denotes the standard orthonormal
basis for $\R^3$. The components of the first fundamental form are the given by
\begin{subequations}
\begin{equation}
[G_{\alpha\beta}]:= [{\bf R}_{,\alpha}\cdot {\bf R}_{,\beta}]=\left(\begin{array}{cc}
1& 0 \\
0&\sin^2(\theta^{1})
\end{array}\right),
\label{eq:matrix_metric}
\end{equation}
and thus
\begin{equation}
G=\sin^2\theta^1.
\label{eq:G_sin2}
\end{equation}
\label{eq:metric_G}%
\end{subequations}

We assume that the current configuration, $\Sigma=\f(S^2)$, is a
smooth, closed surface of genus zero and total area $4\pi$. As in
(\ref{eq:solution_Y}), we consider the convected parametrization
\begin{equation}
{\bf Y}(\theta^1,\theta^2)=\f({\bf R}(\theta^1,\theta^2)).
\label{eq:Y}
\end{equation}
By virtue of the Riemann-Roch theorem \cite{Jost_CRS}, we know that $\Sigma$ is
conformally equivalent to $S^2$, and without loss of generality, we assume that 
$\f(\cdot)$ is chosen in (\ref{eq:Y}) such that 
\begin{equation}
g_{\alpha\beta}:={\bf Y}_{,\alpha}\cdot {\bf Y}_{,\beta} = \lambda(\theta^1,\theta^2)G_{\alpha\beta},
\label{eq:metric}
\end{equation}
where $\lambda(\cdot): \Omega\to \R^+$ is smooth, positive and
bounded. We remark that (\ref{eq:metric}) follows provided that
$\f:S^2\to\Sigma$ is a harmonic map, cf. \cite{Jost_CRS}, \cite{LWYGL}. In
view of (\ref{eq:metric_G}) and (\ref{eq:metric}), we have
\begin{equation}
g=\lambda^2G.
\label{eq:g=lambda2G}
\end{equation}

Next we define new coordinates as follows:
\begin{equation}
\Phi_1(\theta^1):=\arccos\Big(1-\frac{1}{2\pi}\int_0^{\theta^1}\int_0^{2\pi}\sqrt{g}(\sigma,\tau)\;d\tau d\sigma\Big),
\label{eq:phi_1}
\end{equation}
where, using (\ref{eq:G_sin2}), (\ref{eq:metric}) and (\ref{eq:g=lambda2G}),  we have
\begin{equation}
\sqrt{g}=\lambda(\theta^1,\theta^2)\sin\theta^1.
\label{eq:sqrt_g}
\end{equation}
Note that
\begin{subequations}
\begin{equation}
\Phi_1:[0,\pi]\to[0,\pi]\text{ is continuous and strictly increasing, with}
\label{eq:phi_1_inc}
\end{equation}
\begin{equation}
\Phi_1(0)=0.
\label{eq:phi_1_0_0}
\end{equation}
\label{eq:phi_1_inc_1_0_0}%
\end{subequations}
Now from the constraint (\ref{eq:EL_G_constraint}), we have
$$\int_\Sigma\;da=\int_0^\pi\int_0^{2\pi} \sqrt{g}(\sigma,\tau)\;d\tau d\sigma=4\pi,$$
which together with (\ref{eq:phi_1}) yields
\begin{equation}
\Phi_1(\pi)=\arccos(-1)=\pi.
\label{eq:phi_1(pi)=pi}
\end{equation}
By virtue of (\ref{eq:phi_1}), we also note that  
\begin{equation}
\frac{d\Phi_1}{d\theta^1}=\frac{\sin\theta^1}{2\pi\sin(\Phi_1(\theta^1))}\int_0^{2\pi}
\lambda(\theta^1,\tau)\;d\tau > 0\text{ on }(0,\pi).
\label{eq:phi1_prime}
\end{equation}

Next we define
\begin{equation}
\Phi_2(\theta^1,\theta^2):=\frac{2\pi\int_0^{\theta^2}\lambda(\theta^1,\tau)\;d\tau}{
\int_0^{2\pi}\lambda(\theta^1,\tau)\;d\tau}.
\label{eq:phi_2}
\end{equation}
Clearly, for each $\theta^1\in[0,\pi]$:
\begin{subequations}
\begin{equation}
\theta^2\mapsto \Phi_2(\theta^1,\theta^2)\text{ is continuous and strictly increasing on }[0,2\pi],
\label{eq:phi_2_cont}
\end{equation}
\begin{equation}
\Phi_2(\theta^1,0)=0,\text{ and}
\label{eq:phi_2_0}
\end{equation}
\begin{equation}
\Phi_2(\theta^1,2\pi)=2\pi.
\end{equation}
\label{eq:phi_2_prop}%
\end{subequations}
From (\ref{eq:phi_2_prop}) we also note that 
\begin{equation}
\frac{\partial\Phi_2}{\partial\theta^2}=\frac{2\pi\lambda(\theta^1,\theta^2)}{
\int_0^{2\pi}\lambda(\theta^1,\tau)\;d\tau} > 0\text{ on }[0,\pi]\times[0,2\pi].
\label{eq:phi2_prime}
\end{equation}

Now define the coordinate change
\begin{equation}
\Phi(\theta^1,\theta^2)=\Big(\Phi_1(\theta^1),\Phi_2(\theta^1,\theta^2)\Big).
\label{eq:phi_def}
\end{equation}
Since $\Phi_1(\cdot)$ is independent of $\theta^2$, we see from
(\ref{eq:phi1_prime}) and (\ref{eq:phi2_prime}), that the Jacobian determinant of the
transformation satisfies
\begin{equation}
\det(D\Phi) = \frac{\lambda(\theta^1,\theta^2)\sin(\theta^1)}{\sin(\Phi_1(\theta^1))}>0\text{ on }\Omega.
\label{eq:detDPhi}
\end{equation}
For our construction that follows, we need:

\begin{lemma}
$\Phi:\bar\Omega\to\bar\Omega$ is a homoeomorphism, and $\Phi:\Omega\to\Omega$ is an 
orientation preserving diffeomorphism.
\end{lemma}
\begin{proof}
By the construction (\ref{eq:phi_1})-(\ref{eq:detDPhi}), it follows that 
\begin{equation}
\Phi\in C(\bar\Omega,\R^2)\cap C^1(\Omega,\R^2),
\label{eq:phi_continuous}
\end{equation}
with positive Jacobian determinant on $\Omega$. Consider the mapping,
$\tilde\Phi_2:\bar\Omega\to\R$, defined by
\begin{equation}
\tilde\Phi_2(\theta^1,\theta^2):=(1-\frac{\theta^1}{\pi})\Phi_2(0,\theta^2)+\frac{\theta^1}{\pi}\Phi_2(\pi,\theta^2),
\label{eq:phi_tilde}
\end{equation}
and further define $\tilde\Phi:\bar\Omega\to \R^2$ via the continuous map
\begin{equation}
\tilde\Phi(\theta^1,\theta^2):=\Big(\Phi_1(\theta^1),\tilde{\Phi}_2(\theta^1,\theta^2)\Big).
\label{eq:phi_tilde_def}
\end{equation}
By virtue (\ref{eq:phi_1_inc_1_0_0}) and (\ref{eq:phi_2_prop}a,b), we note that
$$\Phi_1(\theta^1)=0\iff \theta^1=0,$$ 
and
$$\tilde\Phi(0,\theta_2)\equiv\Phi_2(0,\theta_2)=0\iff \theta_2=0,$$
respectively, \emph{i.e.,} 
\begin{equation}
\tilde\Phi(\cdot)\text{  is injective on }\bar\Omega.
\label{eq:phi_tilde_injective}
\end{equation}

We further claim that
\begin{equation}
\Phi|_{\partial\Omega}=\tilde{\Phi}|_{\partial\Omega}.
\label{eq:phi_boundary}
\end{equation}
Indeed, from (\ref{eq:phi_1_inc_1_0_0}), (\ref{eq:phi_1(pi)=pi}),
(\ref{eq:phi_2_prop}), (\ref{eq:phi_tilde}) and (\ref{eq:phi_tilde_def}), we
find:
\begin{align*}
&\tilde\Phi(0,\theta^2)=\Phi(0,\theta^2)=(0,\Phi_2(0,\theta^2)),&\\
&\tilde\Phi(\pi,\theta^2)=\Phi(\pi,\theta^2)=(\pi,\Phi_2(\pi,\theta^2)),\;\theta^2\in[0,2\pi];&\\
&\tilde\Phi(\theta^1,0)=\Phi(\theta^1,0)=(\Phi_1(\theta^1),0),&\\
&\tilde\Phi(\theta^1,2\pi)=\Phi(\theta^1,2\pi)=(\Phi_1(\theta^1),2\pi),\;\theta^1\in[0,\pi].&\\
\end{align*}
With (\ref{eq:detDPhi}), (\ref{eq:phi_continuous}),
(\ref{eq:phi_tilde_injective}) and (\ref{eq:phi_boundary}) in hand,
the first assertion now follows from a well known argument based on the Brouwer
degree, cf. \cite{ciarlet}. The inverse function theorem then implies that
$\Phi(\cdot)$ is a local $C^2$-diffeomorphism, and thus it is globally so on
$\Omega$.
\end{proof}

We now consider a new parametrization of $\Sigma$ given by
\begin{equation}
{\bf r}(\phi^1,\phi^2):={\bf Y}(\Phi^{-1}(\phi^1,\phi^2)),\;(\phi^1,\phi^2)\in\Omega,
\label{eq:r_phi}
\end{equation}
where ${\bf Y}(\cdot)$ and $\Phi(\cdot)$ are as defined in (\ref{eq:Y}) and (\ref{eq:phi_def}),
respectively. Denoting the components of the first fundamental form as
\begin{equation}
[a_{\alpha\beta}]:=\frac{\partial{\bf r}}{\partial\phi^\alpha}\cdot \frac{\partial{\bf r}}{\partial\phi^\beta},
\end{equation}
then direct differentiation of (\ref{eq:r_phi}), using (\ref{eq:metric}), yields
$$[a_{\alpha\beta}]=D\Phi^{-T}D{\bf Y}^TD{\bf Y}D\Phi^{-1}$$
\begin{equation}
=\lambda D\Phi^{-T}[G_{\alpha\beta}]D\Phi^{-1}.
\label{eq:a_new_ab}
\end{equation}
Taking the determinant of both sides of (\ref{eq:a_new_ab}) leads to
$$a=\frac{\lambda^2 G}{(\det D\Phi)^2},$$
and subsequent use of (\ref{eq:G_sin2}) and (\ref{eq:detDPhi}) then gives
\begin{equation}
a=\sin^2(\phi^1).
\label{eq:a_in_phi}
\end{equation}

We now state:
\begin{theorem}
Suppose that (\ref{eq:Y}) yields a solution of the $GA$ formulation
(\ref{eq:EL_G}) according to (\ref{eq:solution}), such that the closed
surface $\Sigma=\f(S^2)$ has genus zero, and thus (\ref{eq:metric})
holds. Define the diffeomorphism ${\bm \chi}:S^2\to S^2$,
\begin{equation}
{\bm\chi}:={\bf R}\circ{\Phi}^{-1}\circ{\bf R}^{-1},
\label{eq:my_chi}
\end{equation}
where ${\bf R}(\cdot)$ and $\Phi(\cdot)$ are as defined by (\ref{eq:R}) and
(\ref{eq:phi_def}), respectively. Then
\begin{equation}
{\bf u}^* := \Big(\f\circ{\bm \chi},\phi(\f\circ{\bm \chi})\Big),
\label{eq:u_star}
\end{equation}
belonging to the equivalence class $[{\bf u}]$, cf (\ref{eq:equiv_class}),
is a solution of the $LA$ formulation.
\label{thm:equiv_equil}
\end{theorem}
\begin{proof}
It is enough to show that (\ref{eq:J_equals_1}) is satisfied by
$\f^*:=\f\circ\chi$. Now (\ref{eq:Y}), (\ref{eq:r_phi}) and
(\ref{eq:my_chi}) yield
\begin{equation}
{\bf r}(\phi^1,\phi^2)=\f^*({\bf R}(\phi^1,\phi^2)),\;(\phi^1,\phi^2)\in\Omega,
\end{equation}
and direct differentiation leads to
\begin{equation}
[a_{\alpha\beta}]=D{\bf R}^TD\f^{*T}D\f^*D{\bf R}.
\label{eq:a_proof}
\end{equation}
Taking the determinant of both sides of (\ref{eq:a_proof}) and employing
(\ref{eq:G_sin2}) and (\ref{eq:a_in_phi}), we see that
$$\det D\f^{*T}D\f^{*}\equiv 1.$$
\end{proof}
\section{Equivalence of Second Variation Condition}
\label{sec:SV_stab_constr}

Our next goal is to show that the second variation conditions
including admissibility of variations for both formulations are
equivalent. Since the problem involves constraints, specifically all
smooth variations $(\v,w,\psi)$ (cf. Section
2\ref{sc:LA_formulation}) must satisfy the linearized constraint
equations, which define admissibility. The two variants of the area
constraint give us two seemingly different criteria for admissibility.

In the $LA$ formulation, the linearization of the local area
constraint (\ref{eq:LA_constr}) is given by
\begin{equation}
(\nabla\cdot\v-2Hw)=0\text{ on }\Sigma,
\label{eq:lin_local_area}
\end{equation}
for all smooth variations $\v$, $w$, cf. (\ref{eq:def_var_cur}).

To see this, write $\f(\x)\to\f(\x)+\alpha{\bm\eta}(\x)$ for $\alpha$
sufficiently small with ${\bm \eta}(\cdot)$ smooth (as in
(\ref{eq:f->f+eta})), and consider
\begin{equation}
\det\Big(D\f+\alpha D{\bm \eta}\Big)^T\Big(D\f+\alpha D{\bm\eta}\Big)\equiv 1,
\label{eq:D(f+eta)^TD(f+eta)=1}
\end{equation}
cf. (\ref{eq:LA_constr}). Differentiating
(\ref{eq:D(f+eta)^TD(f+eta)=1}) with respect to $\alpha$ and then
evaluation the result at $\alpha=0$ yields
\begin{equation}
D\f^{-1}D\f^{-T}\cdot\Big(D\f^TD{\bm\eta}+D{\bm\eta}^TD\f\Big)\equiv 0.
\label{eq:lin_Df_J=1}
\end{equation}
From (\ref{eq:f->f+eta}), (\ref{eq:def_var_cur}) and the chain rule we deduce
\begin{equation}
D{\bm\eta}=\Big(D_{\y}\v+\n\otimes\nabla w+w\nabla\n\Big)D\f,
\label{eq:Deta=Dv+bla}
\end{equation}
where the subscript in (\ref{eq:Deta=Dv+bla}) is meant to emphasize the
total derivative with respect to the spatial variable $\y$. Substituting 
(\ref{eq:Deta=Dv+bla}) into (\ref{eq:lin_Df_J=1}) then leads to 
\begin{equation}
\nabla\cdot\v+w\nabla\cdot \n=0,
\end{equation}
where we have used the fact that $\n\cdot\nabla w\equiv 0$. Finally
the identity $\nabla\cdot\n=-2H$ gives (\ref{eq:lin_local_area}).

On the other hand, in the $GA$ formulation, we obtain the following
linearization for the global area constraint (\ref{eq:GA_constr}):
\begin{equation} 
\int_{\Sigma}(\nabla\cdot\v - 2Hw)\;da = 0.  
\end{equation} 
From the divergence theorem, this becomes
\begin{equation} 
\int_{\Sigma}Hw\;da=0.
\label{eq:lin_global_area}
\end{equation}
The linearization of the concentration equation (\ref{eq:conc_constr}) for
both formulations is
\begin{equation}
\int_\Sigma \big(\psi-2H\phi w\Big)\;da = 0.
\label{eq:lin_conc_constr}
\end{equation}

Next we determine expressions for the second variation about an
equilibrium configuration $\Sigma$. First consider the $LA$
formulation. We write the energy (\ref{eq:energy_LA}) as,
\begin{equation}
\E_{LA} = \int_\Sigma {\mathcal F}\;da + \int_{S^2}\gamma_L(\x)(J-1)\;dA,
\end{equation}
where $\mathcal{F}$ contains terms that are common to both the
formulations, \emph{viz.,} $W$, $\phi$, $p$, \emph{etc}. The second
integral in the equation above accounts for the local area
constraint. The first variation can be abstractly written as,
\begin{equation}
\delta\E_{LA} = \int_\Sigma {\mathcal L}[\v,w,\psi] \;da +
\int_\Sigma\gamma_L(\nabla\cdot\v-2Hw)\;da + \int_{S^2}\nu_L(\x)(J-1)\;dA, 
\end{equation}
where $\mathcal{L}$ is a linear operator on the variations
$(\v,w,\psi)$ and $\nu_L(\x)$ is the variation in the Lagrange
multiplier field $\gamma_L$. The second variation then takes the form
\begin{equation}
\delta^2\E_{LA} = \int_\Sigma {\mathcal B}[\v,w,\psi] \;da +
\int_\Sigma\gamma_L(\nabla\cdot\v-2Hw)^2\;da + \int_\Sigma 2\nu_L (\nabla\cdot\v-2Hw)\;da,
\end{equation}
where $\mathcal{B}$ is a bilinear operator on the variations
$(\v,w,\psi)$. By admissibility (\ref{eq:lin_local_area}), the last integral in the
previous equation vanishes. Therefore,
\begin{equation}
\delta^2\E_{LA} =\int_\Sigma \mathcal{B}[\v,w,\psi] + \gamma_L(\nabla\cdot\v - 2Hw)^2\;da.
\label{eq:LA_SV}
\end{equation}
In fact, the second integral in (\ref{eq:LA_SV}) also vanishes, but it is
convenient to keep it for now.

Similarly, for the $GA$ formulation, we obtain the following
expression for the second variation:
\begin{equation}
\delta^2{\E}_{GA} = \int_\Sigma \mathcal{B}[\v,w,\psi]\;da + \gamma_G\int_\Sigma(\nabla\cdot\v-2Hw)^2\;da - 4\nu_G \int_\Sigma Hw\;da.
\end{equation}
Using the admissibility condition (\ref{eq:lin_global_area}), the last
integral may be dropped and the previous equation simplifies to
\begin{equation}
\delta^2{\E}_{GA} = \int_\Sigma \mathcal{B}[\v,w,\psi]\;da +\gamma_G \int_\Sigma(\nabla\cdot\v-2Hw)^2\;da.
\label{eq:GA_SV}
\end{equation}
Recall (cf. Theorem \ref{thm:equiv_equil}) that for any equilibrium
solution of the $GA$ formulation we can find an equivalent
representative solution of the $LA$ formulation, and
$\gamma_L=\gamma_G$ for such solutions. Accordingly we conclude that
at such an equilibrium,
\begin{equation}
\delta^2\E_{LA} = \delta^2\E_{GA}.
\label{eq:GA_SV=LA_SV}
\end{equation}
In other words, the only difference between the second-variation
conditions for the two formulations at an equilibrium is the
apparently different admissibility conditions
(\ref{eq:lin_local_area}) and (\ref{eq:lin_global_area}).

We now show that these are, in fact, the same. We first consider the $1$-form, $v^\flat$ associated with
$\v$ on $\Sigma$, defined by
$$v^\flat := v_\beta dx^\beta,$$ where
$v_\beta=g_{\alpha\beta}v^\alpha$. Using the Hodge decomposition for
the manifold $\Sigma$ \cite{warner}, we write
\begin{equation}
v^\flat = d\sigma + \delta\tau + \eta,
\label{eq:hodge_decomp}
\end{equation}
where $\sigma$ is a 0-form (function), $\tau$ is a 2-form, $\eta$ is a
1-harmonic form, $d(\cdot)$ is the exterior derivative operator and
$\delta(\cdot)$ is the codifferential operator. For any vector field
$\v$ and scalar field $\sigma$ on a 2-manifold, we note the following
standard identities:
\begin{subequations}
\begin{equation}
\nabla\cdot\v=-\delta v^\flat,
\end{equation}
\begin{equation}
\Delta\sigma=-\delta d\sigma.
\label{eq:L-B}
\end{equation}
\end{subequations}
Then using the decomposition (\ref{eq:hodge_decomp}), we write 
$$\nabla\cdot\v=-\delta(d\sigma+\delta\tau+\eta)$$
\begin{equation}
=-\delta d\sigma =\Delta \sigma,
\label{eq:div_v_is_Lap_sigma}
\end{equation}
where we have used the facts that $\delta\eta=0$ and $\delta^2=0$
(cf. \cite{warner}), and (\ref{eq:L-B}). By virtue of
(\ref{eq:div_v_is_Lap_sigma}) we may rewrite the linearized local area
constraint (\ref{eq:lin_local_area}) as follows,
\begin{equation}
\Delta\sigma=2Hw.
\label{eq:poisson}
\end{equation}
From the Fredholm alternative for elliptic PDE \cite{evans}, we
conclude that equation (\ref{eq:poisson}) has a solution if and only
if $2Hw$ is orthogonal to the null space of of the adjoint operator
$\Delta$. Together with the self-adjointness of $\Delta$ and the
fact that the only harmonic functions on a compact oriented Riemannian
manifold are constant functions \cite{warner}, we conclude that for
variations $\v$ and $w$,

\begin{equation}
\nabla\cdot\v =2Hw \iff \int_\Sigma 2Hw\;da = 0.
\label{eq:admissible_equivalent}
\end{equation}

We now conclude:
\begin{theorem}
Let the equilibrium configuration $\Sigma$ be a smooth closed surface
of genus zero. Then the second variations expressions for the two
formulations are the same (cf. (\ref{eq:LA_SV}), (\ref{eq:GA_SV}),
(\ref{eq:GA_SV=LA_SV})) and their admissibility conditions
(\ref{eq:lin_local_area}), (\ref{eq:lin_global_area}) are equivalent.
\label{thm:stability_equiv}
\end{theorem}
\noindent
An immediate consequence is:
\begin{corollary}
With $\u$, $\u^*\in[\u]$ as given in Theorem \ref{thm:equiv_equil},
the second-variation conditions for $\u$ as a solution of the $GA$ formulation are 
identical to those for $\u^*$ as a solution of the $LA$ formulation.
\label{cor:stability_equiv}
\end{corollary}
\begin{remark}
With (\ref{eq:GA_SV=LA_SV}) and (\ref{eq:admissible_equivalent}) in
hand, we observe that the second term on the right sides of (\ref{eq:LA_SV})
and (\ref{eq:GA_SV}) both vanish.
\end{remark}
\section{Concluding Remarks}
The reparametrization symmetry of the field equations for the $GA$
formulation is a reflection of the inherent in-plane fluidity of the
model. This leads to a large equivalence class of equilibrium
solutions. When such a solution represents a smooth, closed surface of
genus zero, we demonstrate that there is a member of the equivalence
class that also satisfies the field equations for the $LA$
formulation. In particular, it preserves the local area ratio as a
mapping from the unit sphere to the equilibrium configuration. We then
go on to show that all second-variation conditions for these two
solutions -- one as solution of the $GA$ formulation and the other as
a solution of the $LA$ formulation -- are identical.  

Questions of existence and regularity of solutions and their stability
are not addressed in this work. Some progress along these lines has
been made recently in \cite{Healey_SD}, where a plethora of
symmetry-breaking solutions for a class of phase-field models in the
$GA$ formulation have been obtained. Our results here show that each
solution found in \cite{Healey_SD}, has a representative in the
equivalence class that also satisfies the $LA$ field equations --
including and especially the local area constraint.
\bibliographystyle{sanjay}
\bibliography{mybib}
\end{document}